\documentclass[journal]{IEEEtran}
\usepackage{amsmath,amssymb,bbm}
\usepackage{amsthm}
\usepackage{multicol}
\usepackage{graphicx,epsfig}
\usepackage{subfigure}
\usepackage{epstopdf}
\usepackage{verbatim}
\usepackage{soul}
\usepackage{color}
\usepackage{mathptmx}
\usepackage{mathrsfs}
\usepackage{booktabs}
\usepackage{fancyhdr}
\usepackage{tabularx}
\usepackage{xfrac}
\usepackage{times}
\usepackage{hhline}
\usepackage{multirow}
\usepackage{caption}
\usepackage{cite}

\newtheorem{defin}{Definition}

\newtheorem{prop}{Proposition}
\newtheorem{remark}{Remark}

\begin{document}

\title{Dynamic Security Assessment of Small-Signal Stability for Power Systems using Windowed Online Gaussian Process}

\author{Chao Zhai \thanks{Chao Zhai is with the School of Automation, China University of Geosciences, Wuhan 430074 China, and with Hubei Key Laboratory of Advanced Control and Intelligent Automation for Complex Systems, Wuhan 430074 China. Email: zhaichao@amss.ac.cn}}

\maketitle

\begin{abstract}
Due to the evolving nature of power grids and model uncertainty, the online stability assessment of electrical power systems is always a challenging problem. This paper aims to provide a theoretical framework for estimating the region of attraction for power systems in real time. By analyzing and learning the measurement data in a given time duration, a windowed online Gaussian process (GP) approach is developed to provide the real-time security assessment and quantify the uncertainty caused by measurement errors and dynamic evolution of power systems. In addition, the theoretical analysis is conducted to ensure the conference level of the estimated region of attraction. Finally, numerical simulations are implemented on a microgrid model with $9$ buses and $3$ generators to validate the proposed approach. The proposed online assessment approach contributes to improving the situational awareness of human operators, thereby taking remedial actions before the emergency.
\end{abstract}


\section{Introduction}
The real-time detection and monitoring of power system states are crucial to the situational awareness and emergency decision of operators. The online calculation of stability margin for power systems contributes to its responsive regulation and corrections against the disruptive  contingency. In practice, the operation of power system is always subject to external disturbances (e.g., load variations and fluctuation of power generation) and internal parameter variations (e.g., the change of branch impedance), which makes it difficult to construct an accurate power system model for the security assessment.

Security assessment of power systems largely relies on the quantification of stability margin and the estimation of possible risks. As a major type of stability for dynamic system, small-signal stability of operating points characterizes the local property of power systems
and it is normally associated with a particular equilibrium. The region of robust small-signal stability can be identified in the system state space via convex optimization \cite{hung16}. Dynamic security assessment (DSA) aims at the evaluation of stability margin according to the dynamic response of power systems subject to major contingencies (e.g., line outages, load shedding and generator tripping, etc). In practice, the dynamic security assessment focuses on the stability of a given operating state determined either from measurement tools (e.g., supervisory control and data acquisition, phasor measurement unit) or constructed in the framework of scenario analysis for planning or operation purposes \cite{huang02}. As mentioned before, the uncertainty of power systems precludes the accurate security assessment of power systems.

The GP approach provides a powerful mathematical tool to quantify the uncertainties of power systems by learning the sampling data with a priori assumption \cite{csat02}. One major obstacle in the application of GP approach is the non-Gaussianity of posterior process, which can be resolved by approximating this non-Gaussian posterior process using a Gaussian one \cite{see00}. The other obstacle results from the size of kernel matrix, which can be handled by the development of sparse approximation techniques \cite{smo00}. Thus, a sparse online GP approach has been proposed by integrating the sparse representation with an online algorithm \cite{csat02}.

Nevertheless, there is still a lack of effective online approaches for the dynamic security assessment of power system subject to major contingencies. For this reason, this work aims to propose a novel DSA scheme for power systems by integrating stability theory with the online sparse GP approach. Thus, this paper centers on the online security assessment of power systems with limited sampling data using the online GP approach.
Compared with existing work \cite{zhai19,mun13,ber16,jone17}, the key contributions of this work lie in
\begin{enumerate}
    \item Construct the region of attraction for a general differential-algebraic equation (DAE) system with the converse Lyapunov function.
    \item Develop an online GP approach for the dynamic security assessment of power systems with limited sampling data in a finite time window.
    \item Propose a security assessment scheme for practical power system with guaranteed confidence level in theory.
\end{enumerate}

The remainder of this paper is organized as follows. Section \ref{sec:roa} introduces the region of attraction for a general DAE system with the aid of the converse Lyapunov theorem. Section \ref{sec:gp} presents an online GP approach for learning the unknown Lyapunov function by using sampling data in a time window. Section \ref{sec:sas} elaborates on the security assessment scheme and presents theoretical results. Numerical simulations are conducted on a microgrid model to validate the proposed approach in Section \ref{sec:sim}. Finally, Section \ref{sec:con} draws a conclusion and discusses future work.

\section{The ROA of DAE System} \label{sec:roa}
Consider the differential-algebraic equation system as follows
\begin{equation}\label{dae}
\begin{split}
\dot{\mathbf{x}}&=\mathbf{f}(\mathbf{x},\mathbf{y}) \\
      \mathbf{0}&=\mathbf{g}(\mathbf{x},\mathbf{y})
\end{split}
\end{equation}
with $\mathbf{x}\in R^n$ and $\mathbf{y}\in R^m$. And the functions $\mathbf{f}: R^n\times R^m\rightarrow R^n$ and $\mathbf{g}: R^n\times R^m\rightarrow R^m$ are twice continuously differentiable in an open connected set $\Omega$. Suppose that there exists an equilibrium point $(\mathbf{x}^{*}, \mathbf{y}^{*})$ in the DAE system (\ref{dae}), and the partial derivative of $\mathbf{g}$ with respect to $\mathbf{y}$ has full rank on an open connected set that contains this equilibrium point \cite{per16}. This guarantees the existence and uniqueness of solutions to the DAE system (\ref{dae}) for any initial points $(\mathbf{x}_0, \mathbf{y}_0)$ that satisfies the algebraic equation $\mathbf{0}=\mathbf{g}(\mathbf{x}_0,\mathbf{y}_0)$ in the connected set \cite{hill90}. Contraction analysis can be employed to construct the ROA of nonlinear DAE systems \cite{nguyen2017contraction}.

For a specific operating point $(\mathbf{x}_0,\mathbf{y}_0)$, consider the region around this point. Let $\mathbf{x}=\mathbf{x}_0+\delta \mathbf{x}$ and $\mathbf{y}=\mathbf{y}_0+\delta \mathbf{y}$, and we have $\dot{\mathbf{x}}=\dot{\delta \mathbf{x}}+\dot{\mathbf{x}}_0$. Then $\mathbf{f}(\mathbf{x},\mathbf{y})$ and $\mathbf{g}(\mathbf{x},\mathbf{y})$ can be approximated by
\begin{equation*}
\begin{split}
\mathbf{f}(\mathbf{x},\mathbf{y})&=\mathbf{f}(\mathbf{x}_0+\delta \mathbf{x},\mathbf{y}_0+\delta \mathbf{y}) \\
&=\mathbf{f}(\mathbf{x}_0, \mathbf{y}_0)+\frac{\partial\mathbf{f}}{\partial\mathbf{x}}\delta \mathbf{x}+\frac{\partial\mathbf{f}}{\partial \mathbf{y}}\delta\mathbf{y}+O(\partial^2\mathbf{f})
\end{split}
\end{equation*}
and
\begin{equation*}
\begin{split}
\mathbf{g}(\mathbf{x},\mathbf{y})&=\mathbf{g}(\mathbf{x}_0+\delta \mathbf{x},\mathbf{y}_0+\delta \mathbf{y})\\
                        &=\mathbf{g}(\mathbf{x}_0, \mathbf{y}_0)+\frac{\partial\mathbf{g}}{\partial\mathbf{x}}\delta \mathbf{x}+\frac{\partial\mathbf{g}}{\partial \mathbf{y}}\delta\mathbf{y}+O(\partial^2\mathbf{g}),
\end{split}
\end{equation*}
where $O(\partial^2\mathbf{f})$ and $O(\partial^2\mathbf{g})$ denote the higher order terms of Taylor series for $\mathbf{f}(\mathbf{x},\mathbf{y})$ and $\mathbf{g}(\mathbf{x},\mathbf{y})$, respectively. By substituting $\dot{\mathbf{x}}_0=\mathbf{f}(\mathbf{x}_0,\mathbf{y}_0)$ and $\mathbf{0}=\mathbf{g}(\mathbf{x}_0,\mathbf{y}_0)$ into the above Taylor series and ignoring the higher order terms, one obtains
\begin{equation}\label{delta_f}
    \dot{\delta \mathbf{x}}=\frac{\partial\mathbf{f}}{\partial\mathbf{x}}\delta \mathbf{x}+\frac{\partial\mathbf{f}}{\partial \mathbf{y}}\delta\mathbf{y}
\end{equation}
and
\begin{equation}\label{delta_g}
    \mathbf{0}=\frac{\partial\mathbf{g}}{\partial\mathbf{x}}\delta \mathbf{x}+\frac{\partial\mathbf{g}}{\partial \mathbf{y}}\delta\mathbf{y}
\end{equation}
Assume that the matrix $\partial\mathbf{g}/\partial\mathbf{y}$ is invertible at the operating point $(\mathbf{x}_0,\mathbf{y}_0)$.
Note that this assumption is valid if power systems operate in the normal condition. By replacing $\delta \mathbf{y}$
in (\ref{delta_f}) with
$$
-\left(\frac{\partial\mathbf{g}}{\partial\mathbf{y}}\right)^{-1}\frac{\partial\mathbf{g}}{\partial\mathbf{x}}\delta \mathbf{x}
$$
from (\ref{delta_g}), one can get
\begin{equation*}
\begin{split}
    \dot{\delta \mathbf{x}}&=\frac{\partial\mathbf{f}}{\partial\mathbf{x}}\delta \mathbf{x}-\frac{\partial\mathbf{f}}{\partial \mathbf{y}}\left(\frac{\partial\mathbf{g}}{\partial\mathbf{y}}\right)^{-1}\frac{\partial\mathbf{g}}{\partial\mathbf{x}}\delta \mathbf{x} \\
    &=\left[\frac{\partial\mathbf{f}}{\partial\mathbf{x}}-\frac{\partial\mathbf{f}}{\partial \mathbf{y}}\left(\frac{\partial\mathbf{g}}{\partial\mathbf{y}}\right)^{-1}\frac{\partial\mathbf{g}}{\partial\mathbf{x}}\right]\delta \mathbf{x}
\end{split}
\end{equation*}
To simplify the mathematical expression, define
\begin{equation}\label{Amatrix}
\mathbf{A}(\mathbf{x},\mathbf{y})=\frac{\partial\mathbf{f}}{\partial\mathbf{x}}-\frac{\partial\mathbf{f}}{\partial \mathbf{y}}\left(\frac{\partial\mathbf{g}}{\partial\mathbf{y}}\right)^{-1}\frac{\partial\mathbf{g}}{\partial\mathbf{x}}
\end{equation}
If $\mathbf{A}(\mathbf{x},\mathbf{y})$ is a Hurwitz matrix at the operating point $(\mathbf{x}_0,\mathbf{y}_0)$, $\delta \mathbf{x}$ converges to $\mathbf{0}$ as time goes to the infinity \cite{asner70}. It follows from the non-singularity of the matrix $\partial\mathbf{g}/\partial\mathbf{y}$ in (\ref{delta_g}) that $\delta \mathbf{y}$ goes to $\mathbf{0}$ as well. This implies that the operating point $(\mathbf{x}_0,\mathbf{y}_0)$ is asymptotically stable, which enables us to estimate the ROA of power grids by constructing converse Lyapunov function \cite{zhai19}.

The regularity of the algebraic equation $\mathbf{0}=\mathbf{g}(\mathbf{x},\mathbf{y})$ is defined in order to ensure the existence and uniqueness of solutions to the DAE system (\ref{dae}) in a connected set \cite{per16}.
\begin{defin}
The algebraic equation $\mathbf{0}=\mathbf{g}(\mathbf{x},\mathbf{y})$ is regular if the Jacobian of $\mathbf{g}(\mathbf{x},\mathbf{y})$ with respect to $\mathbf{y}$ has the full rank on the connected set $\Omega$, that is $rank\left(\nabla_{\mathbf{y}}\mathbf{g}(\mathbf{x},\mathbf{y})\right)=m$, $\forall(\mathbf{x},\mathbf{y})\in\Omega$.
\end{defin}

The regularity of the algebraic equation $\mathbf{0}=\mathbf{g}(\mathbf{x},\mathbf{y})$ allows us to convert the DAE system (\ref{dae}) into an ordinary differential equation (ODE) system. Then the converse Lyapunov theorem can be employed to estimate the value of Lyapunov function without its analytic form \cite{kha96}. For the DAE system (\ref{dae}) that has a stable state trajectory $\mathbf{\phi}(\mathbf{x},t)$, $t\geq0$, a converse Lyapunov function can be constructed \cite{jone17}. Here the stable state trajectory is defined as the trajectory that converges to a stable equilibrium point as time goes to the infinity. The existence and construction of such converse Lyapunov function is presented as follows.
\begin{prop}\label{lya}
Without loss of generality, let the origin be an asymptotically stable equilibrium point for the DAE system (\ref{dae}), where $\mathbf{f}(\mathbf{x},\mathbf{y})$ is locally Lipschitz with respect to $\mathbf{x}$, and $\mathbf{0}=\mathbf{g}(\mathbf{x},\mathbf{y})$ is regular in the connected set $\Omega$ that contains the origin. $S$ is the region of attraction. Then there exist a continuous positive definite function $W(\mathbf{x})$ and a smooth, positive definite function $V(\mathbf{x})$ such that
$$
V(\mathbf{x})=\int_{0}^{\infty}\alpha(\|\mathbf{\phi}(\mathbf{x},t)\|)dt, \quad V(\mathbf{0})=0
$$
and
\begin{equation}\label{eq:pror1}
\begin{split}
\frac{dV(\mathbf{x})}{dt}&=\frac{\partial V(\mathbf{x})}{\partial \mathbf{x}}\mathbf{f}\left(\mathbf{x},\mathbf{Y}(\mathbf{x})\right)\leq -W(\mathbf{x}),~\forall~\mathbf{x}\in S
\end{split}
\end{equation}
with
\begin{equation} \label{eq:pror2}
\frac{d\mathbf{\phi}(\mathbf{x},t)}{dt}=\mathbf{f}\left(\mathbf{\phi}(\mathbf{x},t),\mathbf{Y}(\mathbf{\phi}(\mathbf{x},t))\right)
\end{equation}
and $\mathbf{\phi}(\mathbf{x},0)=\mathbf{x}$, where $\alpha(z)$ is a class $\Gamma$ function defined in Appendix \ref{app:def}.
$\mathbf{Y}(\mathbf{x})$ is the implicit function determined by the equation $\mathbf{0}=\mathbf{g}(\mathbf{x},\mathbf{y})$.
The level set with $c>0$ is given by
$$
\Omega_c=\{\mathbf{x}\in R^n~|~V(\mathbf{x})\leq c\},
$$
which is a compact subset of $S=\{\mathbf{x}\in R^n |\lim_{t\rightarrow{+\infty}}\mathbf{\phi}(\mathbf{x},t)=\mathbf{0}\}$.
\end{prop}
\begin{proof}
See Appendix \ref{app:lemma}.
\end{proof}

\begin{remark}
If the origin is not the stable equilibrium of the DAE system (\ref{dae}), the coordinate transformation $\mathbf{\bar{x}}=\mathbf{x}-\mathbf{x}^{*}$ and $\mathbf{\bar{y}}=\mathbf{y}-\mathbf{y}^{*}$ can be taken to obtain a new DAE system
\begin{equation}\label{dae_new}
\begin{split}
\dot{\mathbf{\bar}{\mathbf{x}}}&=\mathbf{f}(\mathbf{\bar{x}}+\mathbf{x}^{*},\mathbf{\bar{y}}+\mathbf{y}^{*}) \\
      \mathbf{0}&=\mathbf{g}(\mathbf{\bar{x}}+\mathbf{x}^{*},\mathbf{\bar{y}}+\mathbf{y}^{*})
\end{split}
\end{equation}
where $(\mathbf{x}^{*},\mathbf{y}^{*})$ denotes a stable equilibrium point of the DAE system (\ref{dae}). In this way, Proposition \ref{lya} can apply to the DAE system (\ref{dae_new}) with the origin being its stable equilibrium point.
\end{remark}
In Proposition \ref{lya}, the differential equation $\dot{\mathbf{x}} = \mathbf{f}(\mathbf{x},\mathbf{y})$ normally represents swing equation of generators in power systems \cite{kun94}. The algebraic equation $\mathbf{0}=\mathbf{g}(\mathbf{x},\mathbf{y})$ characterizes the power flow distribution. Inequality \eqref{eq:pror1} indicates that the Lyapunov function $V(\mathbf{x})$ decays over time, while Equation \eqref{eq:pror2} defines a stable trajectory $\mathbf{\phi}(\mathbf{x}, t)$ with the initial state $\mathbf{x}$. For the Lyapunov function $V(\mathbf{x})$ constructed in Proposition \ref{lya}, it is feasible to approach the real ROA by enlarging the level set.
This property may not hold for any analytical Lyapunov function. This advantage enables us to obtain a better ROA by collecting more sampling points in order to enlarge the level sets. In practice, the converse Lyapunov function $V(\mathbf{x})$ proposed in Proposition \ref{lya} can be estimated by \cite{zhai19}
\begin{equation}\label{Vest}
    \hat{V}(\mathbf{x})=\sum_{i=1}^{n}\alpha(\|\mathbf{\phi}(\mathbf{x},t_i)\|)\Delta t,
\end{equation}
where $\Delta t$ denotes the sampling time interval and $t_i=(i-1)\Delta t$, $i\in \{1,2,...,n\}$. While $\hat{V}(\mathbf{x})$ can be calculated directly using (\ref{Vest}) and the sampling data, the analytical Lyapunov function $V(\mathbf{x})$ is unknown. This work aims to learn this unknown Lyapunov function $V(\mathbf{x})$ online by capitalizing on the discrete sampling data and values of converse Lyapunov function in a time window. By treating $V(\mathbf{x})$ and $\hat{V}(\mathbf{x})$ as a GP and its measurement, respectively, the estimation error $V(\mathbf{x})-\hat{V}(\mathbf{x})$ can be regarded as the measurement noise. This enables us to learn the unknown Lyapunov function $V(\mathbf{x})$ using the online GP approach.

\section{The Windowed Online GP} \label{sec:gp}
\begin{figure*}\centering
\scalebox{0.15}[0.15]{\includegraphics{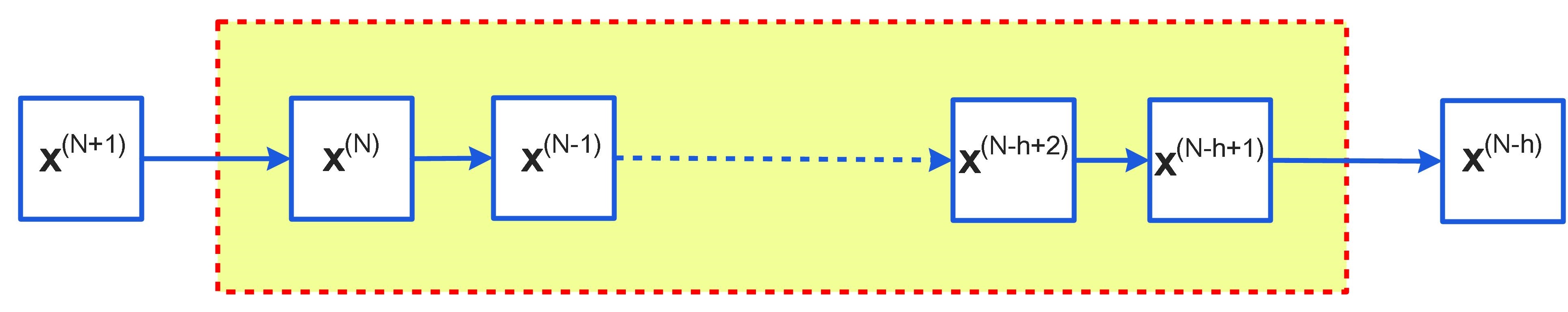}}
\caption{\label{winh} The time window with $h$ sampling points at the $N$-th sampling step. The red dashed rectangle refers to the time window, and it includes $h$ sampling points, which are denoted by the squares with blue boundaries. At each sampling step, the newest sampling point enters the time window, while the oldest one is removed from the window. For example, $\mathbf{x}^{(N)}$ enters the time window $W^N_h$ and $\mathbf{x}^{(N-h)}$ goes out of it at the $N$-th sampling step.}
\end{figure*}
The above section introduces the approximated Lyapunov function $\hat{V}(\mathbf{x})$ and the unknown Lyapunov function $V(\mathbf{x})$. In this section, we propose the windowed online GP approach for learning the unknown Lyapunov function.

\subsection{GP regression}
Normally, a general GP regression requires a prior distribution of unknown functions specified by a mean function, a covariance function, and the probability of the observations and sampling data to obtain the posterior distribution. Without loss of generality, we consider the unknown Lyapunov function $V(\mathbf{x})$ as a GP, which can be sequentially measured by $y^{(i)}=V(\mathbf{x}^{(i)})+\epsilon$, $i\in Z^{+}$, where $y^{(i)}$ refers to the observed function value for the input $\mathbf{x}^{(i)}$ at the $i$-th sampling step, and the measurement noise $\epsilon$ is zero-mean, independent and bounded by $\sigma$. With the GP approach, we can obtain the posterior distribution over $V(\mathbf{x})$ by using sampling data in the training set. By regarding the values of $V(\mathbf{x})$ as random variables, any finite collection of them is multivariate distributed in an overall consistent way. The unknown Lyapunov function $V(\mathbf{x})$ can be approximated by a GP. Note that the covariance or kernel function $k(\mathbf{x},\mathbf{x}')$ encodes the smoothness property of $V(\mathbf{x})$ from the GP. Essentially, the estimated value of Lyapunov function at a sampling point can be regarded as one observation of GP
$$
\hat{V}(\mathbf{x}^{(i)})=V(\mathbf{x}^{(i)})+\epsilon,  \quad i\in Z^{+}
$$
with $\epsilon\sim N(0,\sigma^2)$. Then the first $N$ observations form the vector $\hat{\mathbf{V}}_N=[\hat{V}(\mathbf{x}^{(1)}),...,\hat{V}(\mathbf{x}^{(N)})]^T$. There are the analytic formulas for mean $\mu^N(\mathbf{x})$, covariance $k^N(\mathbf{x},\mathbf{x}')$ and standard deviation $\sigma^N(\mathbf{x})$ of the posterior distribution
as follows \cite{sri12}
\begin{equation}\label{mu_sig}
    \begin{split}
        \mu^N(\mathbf{x})&=\mathbf{k}^N(\mathbf{x})^T(\mathbf{K}^N+\sigma^2\mathbf{I}_N)^{-1}\hat{\mathbf{V}}_N \\
       k^N(\mathbf{x},\mathbf{x}')&=k(\mathbf{x},\mathbf{x}')-\mathbf{k}^N(\mathbf{x})^T(\mathbf{K}^N+\sigma^2\mathbf{I}_N)^{-1}\mathbf{k}^N(\mathbf{x}') \\
       \sigma^N(\mathbf{x})&=\left[k^N(\mathbf{x},\mathbf{x})\right]^{\frac{1}{2}}
    \end{split}
\end{equation}
where $\mathbf{k}^N(\mathbf{x})=[k(\mathbf{x}^{(1)},\mathbf{x}),...,k(\mathbf{x}^{(N)},\mathbf{x})]^T$ and $\mathbf{K}^N$ is the positive definite kernel matrix $[k(\mathbf{x},\mathbf{x}')]$ with $\mathbf{x},\mathbf{x}'\in W_{\infty}^N$. And the set $W_{\infty}^N$ is given by
$W_{\infty}^N=\{\mathbf{x}^{(1)},\mathbf{x}^{(2)},...,\mathbf{x}^{(N)}\}$. Although it is convenient to adopt the formula (\ref{mu_sig}) for the inference, non-Gaussian of the posterior distribution and the size of matrix $\mathbf{K}^N$ preclude its direct application \cite{csat02}. In practice, the proper selection of the width of time window can resolve the problem of matrix size due to large datasets. The non-Gaussian of the posterior distribution can be handled by the approximation of Gaussian ones. Thus, an online algorithm is developed to integrate the windowed GP with the approximation of non-Gaussian posterior process.

\subsection{Windowed online GP}
By approximating the posterior in the sense of the Kullback-Leibler divergence, the online GP is given by \cite{csat02}
\begin{equation}\label{online_GP}
\begin{split}
\mu^{N}(\mathbf{x})&=\mu^{N-1}(\mathbf{x})+q^{N}\cdot k^{N-1}(\mathbf{x},\mathbf{x}^{(N)}) \\
k^{N}(\mathbf{x},\mathbf{x}')&=k^{N-1}(\mathbf{x},\mathbf{x}')+r^{N}\cdot k^{N-1}(\mathbf{x},\mathbf{x}^{(N)})k^{N-1}(\mathbf{x}^{(N)},\mathbf{x}') \\
\sigma^{N}(\mathbf{x})&=\left[k^{N}(\mathbf{x},\mathbf{x})\right]^{\frac{1}{2}}
\end{split}
\end{equation}
where $q^{N}$ and $r^{N}$ are updated as follows
\begin{equation}
\begin{split}
q^{N}&=\frac{\partial}{\partial E[V^{N}]_{N-1}}\ln{E[p(\hat{V}^{N}|V^{N})]_{N-1}} \\
r^{N}&=\frac{\partial^2}{\partial E[V^{N}]^2_{N-1}}\ln{E[p(\hat{V}^{N}|V^{N})]_{N-1}}
\end{split}
\end{equation}
with $V^{N}=V(\mathbf{x}^{(N)})$ and $\hat{V}^{N}=\hat{V}(\mathbf{x}^{(N)})$. Due to limited resources of computation and the evolution of power systems, it is necessary to remove the old sampling points and add the latest ones into the sampling set when the number of sampling points is larger than a given threshold. Thus, a time window is introduced to include the latest $h$ sampling points for GP learning.
In order to allow for the evolution of power systems and relieve the computational burdens, a time window is constructed to take into account the latest $h$ sampling points for the GP learning (see Fig. \ref{winh}). Thus, a set of sampling data is defined as
$$
W^N_h=\left\{\mathbf{x}^{(N-h+1)},\mathbf{x}^{(N-h+2)},...,\mathbf{x}^{(N-1)},\mathbf{x}^{(N)}\right\}
$$
to include the latest $h$ sampling points at the $N$-th sampling step. By unfolding the recursion steps in (\ref{online_GP}), the parametrization of approximate posterior GP can be obtained as follow
\begin{equation}\label{par_GP}
\begin{split}
\mu_h^{N}(\mathbf{x})&=(\mathbf{\alpha}^{N})^T\mathbf{k}_h^{N}(\mathbf{x}) \\
k_h^{N}(\mathbf{x},\mathbf{x}')&=k(\mathbf{x},\mathbf{x}')+\sum_{\mathbf{x}^{(i)},\mathbf{x}^{(j)}\in W_h^{N}}C^{N}_{ij}\cdot k(\mathbf{x},\mathbf{x}^{(i)})k(\mathbf{x}^{(j)},\mathbf{x}') \\
\sigma_h^{N}(\mathbf{x})&=\left[k_h^{N}(\mathbf{x},\mathbf{x})\right]^{\frac{1}{2}}
\end{split}
\end{equation}
where $\mathbf{k}_h^{N}(\mathbf{x})=[k(\mathbf{x}^{(N-h+1)},\mathbf{x}),...,k(\mathbf{x}^{(N)},\mathbf{x})]^T$.
Here $W_h^{N}$ denotes the set of sampling points in the time window at the $N$-th iteration. Let $\mathbf{K}_h^{N}=\{k(\mathbf{x},\mathbf{x}')\}\in R^{h\times h}$ represent the kernel matrix with $\mathbf{x},\mathbf{x}'\in W_h^{N}$.
Then an operator ${R}$ is introduced to update $\mathbf{K}_h^{N}$ as follows
\begin{equation}\label{Kmatrix}
\mathbf{K}_h^{N}={R}(\mathbf{K}_h^{N-1})+[\mathbf{0}_{h\times(h-1)},T(\mathbf{k}_h^{N})]+[\mathbf{0}_{h\times (h-1)},T(\mathbf{k}_h^{N})]^T
\end{equation}
with
$$
\mathbf{k}_h^{N}=[k(\mathbf{x}^{(N-h+1)},\mathbf{x}^{(N)}),...,k(\mathbf{x}^{(N-1)},\mathbf{x}^{(N)})]^T.
$$
The definition of the operator $R$ is presented in Appendix \ref{operator}, and the coefficients $\mathbf{\alpha}^{N}$ and $\mathbf{C}^N=\{C^N_{ij}\}$ with $\mathbf{x}^{(i)},\mathbf{x}^{(j)}\in W_h^{N}$ are obtained according to the updating rule:
\begin{equation}\label{par_update}
\begin{split}
\mathbf{k}_h^{i}&=\Lambda(\mathbf{K}_{h}^{N}\mathbf{e}^i_h) \\
\mathbf{s}^{i}&={T}(\mathbf{C}^{i-1}\mathbf{k}_h^{i})+\mathbf{e}^{i+h-N} \\
\mathbf{\alpha}^{i}&={T}(\mathbf{\alpha}^{i-1})+q^{i}\cdot\mathbf{s}^{i} \\
\mathbf{C}^{i}&={U}(\mathbf{C}^{i-1})+r^{i}\cdot\mathbf{s}^{i}(\mathbf{s}^{i})^T
\end{split}
\end{equation}
where $\mathbf{k}_h^{i}=[k(\mathbf{x}^{(i-h+1)},\mathbf{x}^{(i)}),...,k(\mathbf{x}^{(i-1)},\mathbf{x}^{(i)})]^T$ and the iteration number $i$ increases sequentially from $(N-h+1)$ to $N$. Moreover, $\mathbf{e}^i_h$ represents a $h$ dimensional unit vector with the $i$-th element being $1$. And $\Lambda$ is an operator that can construct a vector by extracting the first $(i-1)$ elements from a given vector. ${T}$ and ${U}$ are two operators that extend the vector and matrix by one dimension, respectively. Specifically, ${T}$ adds zero at the end of the vector, and ${U}$ appends zeros to the last row and column of the matrix. In addition, $\mathbf{e}^{i+h-N}$ refers to the $(i+h-N)$-th unit vector.

\section{Security Assessment Scheme}\label{sec:sas}

\begin{figure*}\centering
\scalebox{0.075}[0.075]{\includegraphics{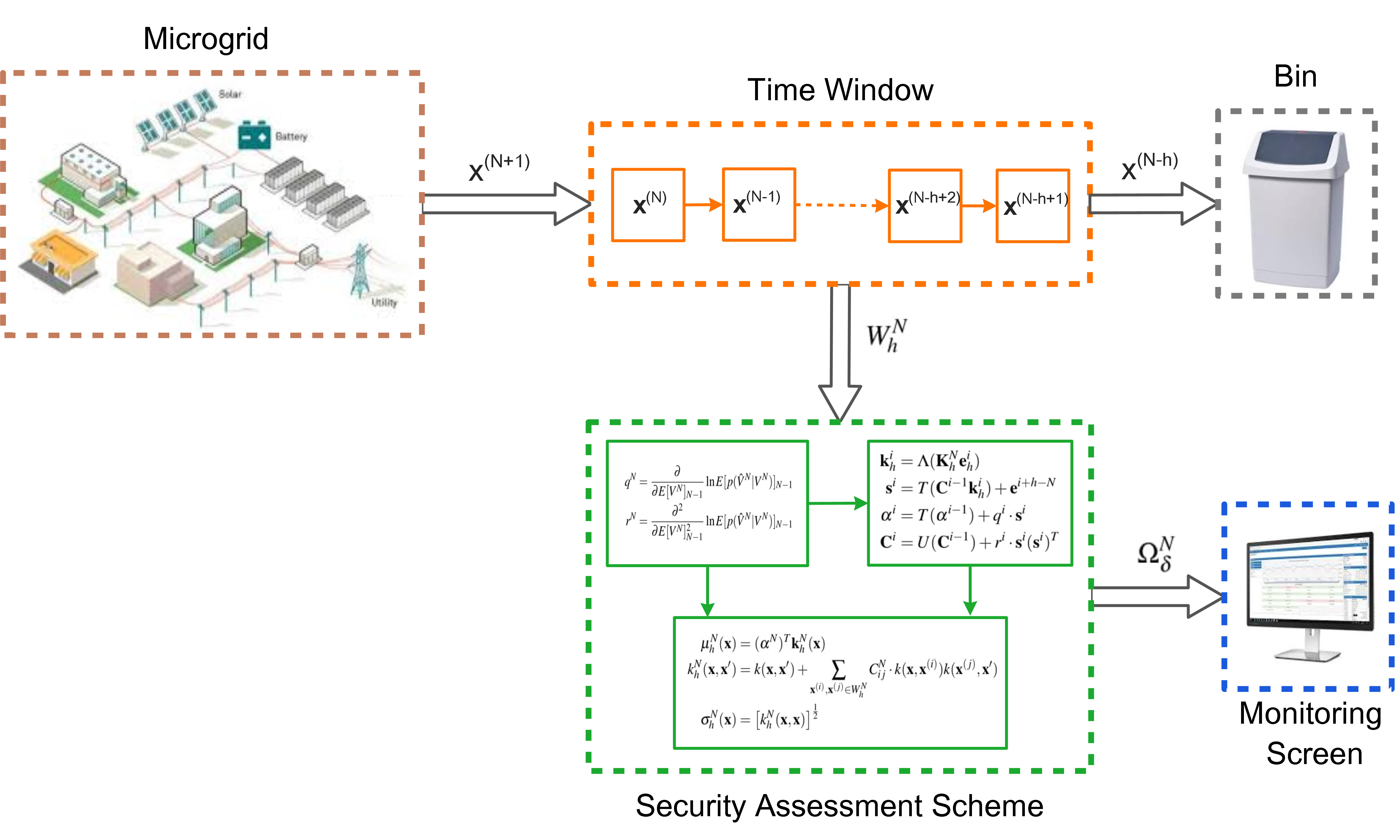}}
\caption{\label{flow} The information flow of the online security assessment in a microgrid.}
\end{figure*}

For a given value of $\delta \in (0,1)$ and the sampling domain $X\in R^n$, which is a subset of the operating state space, our goal is to estimate the region of attraction online, wherein each point converges to the origin with the probability of $\delta$ at least.
Thus, a security assessment scheme is developed in Table \ref{tab:sta} to select the sampling points for enlarging the ROA with a guaranteed confidence level. Specifically, a sampling point $\mathbf{x}^{(N)}$ is selected in $X$ at the $N$-th sampling step by searching for the maxima of $\mu_h^{N-1}(\mathbf{x})+\beta_{\delta}\cdot\sigma_h^{N-1}(\mathbf{x})$, where the term $\mu_h^{N-1}(\mathbf{x})$ helps to enlarge the level set of Lyapunov function and the term $\sigma_h^{N-1}(\mathbf{x})$ allows to reduce the uncertainty of sampling region.

Essentially, the sampling rule aims to reconcile the trade-off between the exploitation for enlarging the ROA and the exploration for reducing the uncertainty of sampling region. Let $\mathbf{x}^{(N)}$ serve as the sampling point of the DAE system (\ref{dae}), and it enables us to generate a state trajectory $\mathbf{\phi}(\mathbf{x}^{(N)},t)$, $t\geq0$. If this state trajectory can converge to the origin, the point $\mathbf{x}^{(N)}$ is called as a stable sampling point. Then the value of Lyapunov function at $\mathbf{x}^{(N)}$ is estimated by $\hat{V}(\mathbf{x}^{(N)})$ with (\ref{Vest}). The time window $W_h^N$ is updated by removing an old sampling point $x^{(N-h)}$ and including a new one $\mathbf{x}^N$. By choosing $\{(\mathbf{x}^{(N-h+1)},\hat{V}(\mathbf{x}^{(N-h+1)})),...,(\mathbf{x}^{(N)},\hat{V}(\mathbf{x}^{(N)}))\}$ as the training set, $\mu_h^N(\mathbf{x})$ and $\sigma_h^N(\mathbf{x})$ for the unknown Lyapunov function $V(\mathbf{x})$ can be updated according to (\ref{par_GP}). Note that the initial training set is assigned as $\{(\mathbf{0},0),(\mathbf{0},0),...,(\mathbf{0},0)\}$. If the state trajectory $\mathbf{\phi}(\mathbf{x}^{(N)},t)$, $t\geq0$ fails to converge to the origin, the sampling point $\mathbf{x}^{(N)}$ will be reselected according to the sampling rule. Finally, the region of attraction can be estimated online by constructing $\Omega^N_{\delta}$ with (\ref{roa}).

\begin{table}
\caption{\label{tab:sta} Security Assessment Scheme}
\begin{center}
\begin{tabular}{lcl} \hline
 \textbf{Input:}  $X\in R^n$, $\delta$, $\xi$, $t_n$, $h$, $\mu_0$, $\sigma_0$, $k(\mathbf{x},\mathbf{x}')$, $N=1$ \\
 \textbf{Output:} $W_h^N$, $\mu_h^N(\mathbf{x})$, $\sigma_h^N(\mathbf{x})$, $\Omega_{\delta}^N$ \\ \hline
  1: ~~\textbf{while} $(1)$ \\
  2:~~~~~~~ Choose $\mathbf{x}^{(N)}=\arg\max_{\mathbf{x}\in X}\left[\mu_h^{N-1}(\mathbf{x})+\beta_{\delta}\cdot\sigma_h^{N-1}(\mathbf{x})\right]$ \\
  3:~~~~~~~ Generate $\mathbf{\phi}(\mathbf{x}^{(N)},t)$, $t\geq0$ with (\ref{dae}) \\
  4:~~~~~~~ \textbf{if} ($\|\mathbf{\phi}(\mathbf{x}^{(N)},t_n)\|<\xi$) \\
  5:~~~~~~~~~~~ Sample $\hat{V}(\mathbf{x}^{(N)})=V(\mathbf{x}^{(N)})+\epsilon$ with (\ref{Vest}) \\
  6:~~~~~~~~~~~ Update $W_h^N=W_h^{N-1}\cup\{\mathbf{x}^{(N)}\}\setminus\{\mathbf{x}^{(N-h)}\}$ \\
  7:~~~~~~~~~~~ Compute $\mathbf{\alpha}^N$ and $\mathbf{C}^N$ with (\ref{Kmatrix}) and (\ref{par_update}) \\
  8:~~~~~~~~~~~ Update $\mu_h^N(\mathbf{x})$ and $\sigma_h^N(\mathbf{x})$ with  (\ref{par_GP}) \\
  9:~~~~~~~ \textbf{else} \\
  10:~~~~~~~~~~ Go to Step $2$ \\
  11:~~~~~~ \textbf{end if} \\
  12:~~~~~~~~~~ Construct $\Omega^N_{\delta}$ with (\ref{roa})\\
  13:~~~~~~~~~~ Update $N=N+1$ \\
  14: ~\textbf{end while} \\ \hline
\end{tabular}
\end{center}
\end{table}

If a certified ROA is available, it is sufficient to judge a stable sampling point if the corresponding state trajectory can enter this ROA.
Each state in the ROA is guaranteed to approach the stable equilibrium point as time goes to the infinity. This can reduce the computation time for determining the stable sampling points. Next, we present theoretical results on the construction and evaluation of ROA with a given confidence level by using the Security Assessment Scheme in Table \ref{tab:sta}.

\begin{prop}\label{prop_est}
Let $\delta\in(0,1)$ and $\hat{V}^N_{\max}=\max_{\mathbf{x}^{(i)}\in W_h^N}\hat{V}(\mathbf{x}^{(i)})$. Then the region of attraction of DAE system (\ref{dae}) at the $N$-th sampling step is given by
\begin{equation}\label{roa}
\Omega^N_{\delta}=\left\{\mathbf{x}\in R^n|\mu^N_h(\mathbf{x})+\beta_{\delta}\cdot\sigma^N_h(\mathbf{x})\leq\hat{V}^N_{\max}\right\}
\end{equation}
with the probability of $\delta$ and $\beta_{\delta}=\Phi^{-1}(\frac{1+\delta}{2})$, where $\Phi^{-1}$ is the inverse cumulative distribution function of the standard normal distribution.
\end{prop}

\begin{proof}
See Appendix \ref{prop2}.
\end{proof}

Proposition \ref{prop_est} allows to estimate the ROA of the DAE system (\ref{dae}) with a certain probability. If a certified ROA can be obtained, the proposed scheme can be employed to assess the stability of sampling region outside the certified ROA. Besides the ROA, the proposed scheme is applied to a more general concept of security for power systems. For example, it is feasible to estimate the invariant set of the DAE system (\ref{dae}) for the online security assessment.

\section{Numerical Simulations}\label{sec:sim}

This section presents simulation results using the Security Assessment Scheme in Table \ref{tab:sta} for a micro grid with $9$ buses and $3$ generators. Figure \ref{flow} illustrates the flow of information and data for the online security assessment of a micro grid. First of all, the measurement data are collected from a micro grid using the measurement devices (e.g., phasor measurement unit). Then the sampling data are sent to a time window $W_h^N$ that can accommodate the latest $h$ sampling points, and the out-of-date sampling points are deleted from the time window in order to reflect the system evolution and reduce the computational burden. The $h$ sampling points in the time window are used to train the GP model and implement the security assessment scheme, which can provide the real-time estimation of ROA for the micro grid.

First of all, a mathematical model of microgrid is introduced as follows.

\subsection{Microgrid model}
Consider a $M$-bus microgrid described by a set of differential-algebraic equations as follows \cite{per16}
\begin{equation}\label{eq:dif}
\begin{split}
\dot{\theta}_i    ~~&=~~ \omega_i \\
\dot{\omega}_i ~~&=~~ -(\omega_i-\omega^{*})-K_{P,i}(P_i-P^{*}_i)+u_{P,i} \\
\dot{U}_i  ~~&=~~ -(U_i-U_i^{*})-K_{Q,i}(Q_i-Q^{*}_i), \quad i\in \mathcal{V}_I \\
\end{split}
\end{equation}
and
\begin{equation}\label{eq:alg}
\begin{split}
    0 ~~&=~~ P_i-P^{*}_i \\
    0 ~~&=~~ Q_i-Q^{*}_i, \quad i\in \mathcal{V}_L
\end{split}
\end{equation}
where $\omega^*$ denotes the nominal frequency. $P_i$ and $Q_i$ are the active power and reactive power on node $i$, respectively.
The parameters $K_{P,i}$ and $K_{Q,i}$ are positive constants, and the term $u_{P,i}$ represents a secondary
control input. In addition, $P^{*}_i$, $Q^{*}_i$ and $U_i^{*}$ are the corresponding set points for the variables $P_i$, $Q_i$ and
$U_i$, respectively. The net active power and net reactive power injected at node $i$, $i\in \mathcal{V}=\mathcal{V}_L\cup \mathcal{V}_I$ are given by
$$
P_i=\sum_{j=1}^{M}B_{ij}|U_i||U_j|\sin(\theta_{ij})
$$
and
$$
Q_i=-\sum_{j=1}^{M}B_{ij}|U_i||U_j|\cos(\theta_{ij})
$$
with the assumption that branch admittances are purely inductive. It is demonstrated that the equilibrium determined by set points is asymptotically stable using the following integral control law \cite{per16}
\begin{equation*}
\begin{split}
      \mathbf{u_P}&=\mathbf{\zeta}  \\
\dot{\mathbf{\zeta}}&=-\mathbf{L}\mathbf{\zeta}-\mathbf{K_P}^{-1}(\mathbf{\omega_I}-\mathbf{\omega^{*}_I}),
\end{split}
\end{equation*}
where $\mathbf{K_P}^{-1}=\texttt{diag}(K_{P,i}^{-1})$ is a diagonal matrix with the $i$-th diagonal element $K_{P,i}^{-1}$.
The vector $\mathbf{u_P}=(u_{P,1},u_{P,2},...,u_{P,|\mathcal{V}_I|})^T$ denotes the control input on each generator bus, and $\mathbf{L}$ refers to the Laplacian matrix of a graph related to the power network.


\subsection{Simulation results}
Figure \ref{resp} presents the time response of system states after suffering from the initial disturbances (e.g., branch outage) using the microgrid model. It is observed that the disturbance occurs at $t=1$s, which leads to the fluctuations of power system states (e.g., frequency, phase angle and voltage magnitude). And these power system states converge due to feedback control. Such response trajectories enables us to estimate the value of Lyapunov function for a certain initial state using (\ref{Vest}). The initial state and its corresponding value of Lyapunov function form one element in the training set, and the online GP algorithm is implemented to learn the Lyapunov function in a given time window with the width $h=100$. Note that the spherical radial basis function (RBF) kernels are adopted with the Gaussian likelihood for the GP learning. Figure \ref{mul_roa} shows the estimation of ROA for the frequency of IEEE $9$ Bus System in the four sequential sampling steps. The origin refers to the desired frequency of microgrid, and the states in the yellow region are guaranteed to converge to the origin with the probability larger than 90\%. It is demonstrated that the shape of yellow region changes due to the update of sampling data in the training set, which reflects the evolution of stability margin. By comparing the current operating state and the estimated ROA (e.g., the yellow region), it is convenient for human operators to determine the stability level of power systems and take remedial actions in time during the emergency.


\subsection{Discussions on computational cost}

\begin{figure}\centering
 {\includegraphics[width=0.52\textwidth]{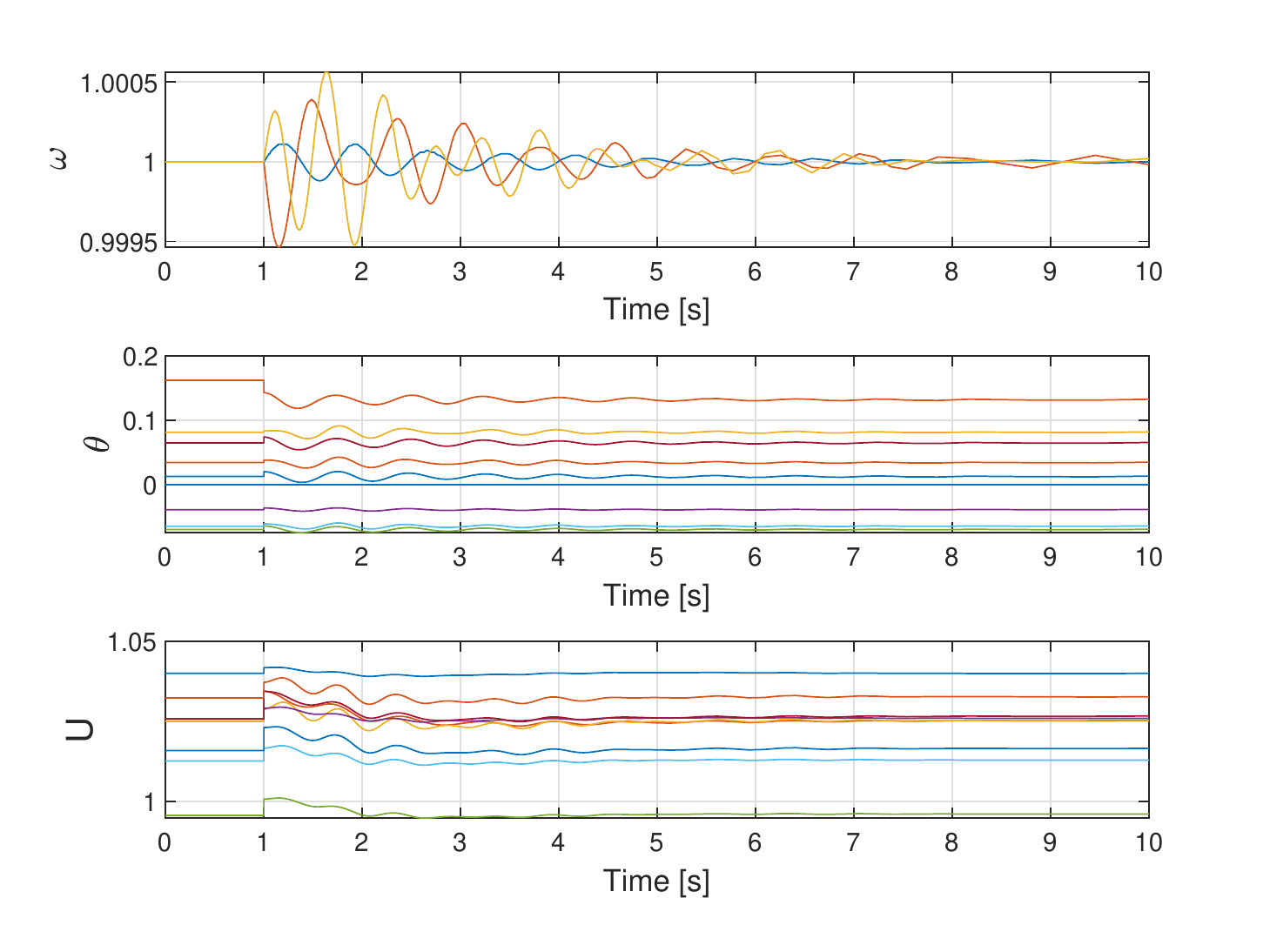}}
 \caption{\label{resp} Time response of power system states $\omega$, $\theta$ and $U$ due to the disturbance on Branch $2$ at time $t=1$s.}
\end{figure}

\begin{figure}\centering
 {\includegraphics[width=0.52\textwidth]{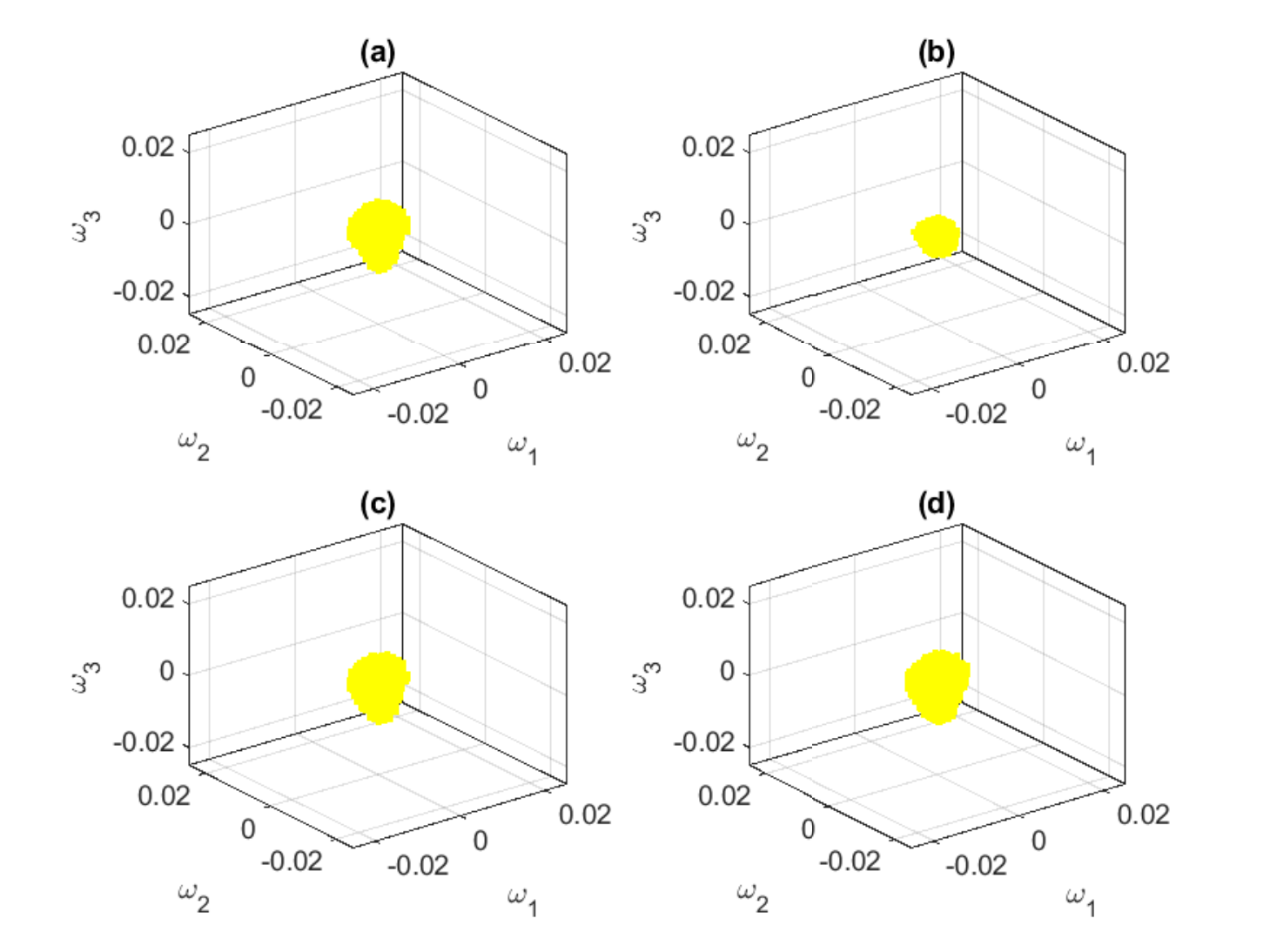}}
 \caption{\label{mul_roa} ROA of power system frequencies in the four sequential sampling steps with the confidence level of 90\%.}
\end{figure}

The computational cost associated with Security Assessment Scheme in Table \ref{tab:sta} mainly results from two factors: the generation of state trajectories and the implementation of online GP algorithm as presented in \eqref{par_GP}, \eqref{Kmatrix} and \eqref{par_update}. Essentially, the former depends on the structure and dimension of DAE system. In fact, the computation burden mainly depends on the number of sampling points $h$ in the time window rather than the dimension of power system dynamics. With Matlab 2017b in the laptop, it takes around $5$ seconds for IEEE $9$-bus system in the simulations. For the large-size sampling data, the sparse representation of GP model can be adopted to overcome the size limitation by reconstructing a sub-sample of the whole sampling data \cite{liu18}.

\section{Conclusions and Future Work}\label{sec:con}
In this paper, we developed an online GP approach to assess the operating states of power systems by learning the sampling data in a finite time window. Numerical simulations were conducted to validate the proposed approach in a microgrid model. The proposed approach enables us to monitor the operation of power grids in real time and provide the timely situational awareness for operators. Future work may include the filtering of redundant sampling data in the time window and the integration of real-time protection schemes for enhancing the resilience of power systems.


\section*{Appendix}

\subsection{The class $\Gamma$ function} \label{app:def}
\begin{defin}\label{def}
The class $\Gamma$ function consists of all continuous functions $\alpha: [0,a)\rightarrow{[0,\infty]}$ which satisfy the following conditions \cite{zhai19}:
\begin{enumerate}
  \item $\forall z>0$, $\alpha(z)\in C^2$.
  \item $\forall z>y\geq0$, $\alpha(z)>\alpha(y)$ and $\alpha(0)=0$.
  \item $\forall z\geq0$, $\exists~m>0$, such that $\alpha(z)\leq z^m$.
\end{enumerate}
\end{defin}

\subsection{Proof of Proposition \ref{lya}}\label{app:lemma}
It follows from the implicit function theorem and the regularity of $\mathbf{g}(\mathbf{x},\mathbf{y})$ with respect to $\mathbf{y}$ that there is a neighborhood $U_{\Omega}\in R^n$ and a unique twice differentiable function $\mathbf{Y}: R^n\rightarrow R^m$ such that $\mathbf{0}=\mathbf{g}(\mathbf{x},\mathbf{Y}(\mathbf{x}))$, $\mathbf{x}\in U_{\Omega}$ \cite{hill90}. Thus, the DAE system (\ref{dae}) reduces to the ODE system $\dot{\mathbf{x}}=\mathbf{f}(\mathbf{x}, \mathbf{Y}(\mathbf{x}))$, $\mathbf{x}\in U_{\Omega}$. Since the origin
is a stable equilibrium point of the DAE system (\ref{dae}), $\mathbf{0}$ is a stable equilibrium point of the ODE $\dot{\Delta \mathbf{x}}=\mathbf{f}\left(\mathbf{x},\mathbf{Y}(\mathbf{x})\right)$. It follows from Theorem $4.17$ in \cite{kha96} that there exist a smooth, positive definite function $V(\mathbf{x})$ and a continuous, positive definite function $W(\mathbf{x})$ such that
\begin{equation*}
\begin{split}
\frac{\partial V(\mathbf{x})}{\partial \mathbf{x}}\mathbf{f}\left(\mathbf{x},\mathbf{Y}(\mathbf{x})\right)\leq -W(\mathbf{x}), ~\forall~\mathbf{x}\in S
\end{split}
\end{equation*}
with the converse Lyapunov function given by
$$
V(\mathbf{x})=\int_{0}^{\infty}\alpha(\|\mathbf{\phi}(\mathbf{x},t)\|)dt, \quad V(\mathbf{0})=0
$$
according to Lemma $1$ in \cite{zhai19}. This completes the proof.

\subsection{The operator ${R}$} \label{operator}
For any matrix $\mathbf{D}\in R^{h\times h}$, the operator $R$ allows to move all elements of $\mathbf{D}$ up along its main diagonal by one slot. Mathematically, it is described as
$$
R(\mathbf{D})=\left[
\begin{array}{cc}
  \mathbf{0}_{h-1} & \mathbf{I}_{h-1} \\
  0 & \mathbf{0}^T_{h-1}
\end{array}
\right]\mathbf{D}\left[
\begin{array}{cc}
  \mathbf{0}^T_{h-1} & 0 \\
  \mathbf{I}_{h-1} & \mathbf{0}_{h-1}
\end{array}
\right]
$$
where $I_{h-1}$ denotes the $(h-1)$ dimensional unit matrix and $\mathbf{0}_{h-1}$ refers to the $(h-1)$ dimensional zero vector.

\subsection{Proof of Proposition \ref{prop_est}} \label{prop2}
For the sampling points in $W^N_h$ and the fixed $\mathbf{x}\in X$, it follows from (\ref{par_GP}) that $V(\mathbf{x})\sim N\left(\mu^N_h(\mathbf{x}),\sigma^N_h(\mathbf{x})\right)$, which leads to
$$
\frac{V(\mathbf{x})-\mu^N_h(\mathbf{x})}{\sigma^N_h(\mathbf{x})}\sim N(0,1).
$$
Thus, for a positive constant $c$, it holds that
\begin{equation*}
\begin{split}
\texttt{Prob}\left\{ \left|\frac{V(\mathbf{x})-\mu^N_h(\mathbf{x})}{\sigma^N_h(\mathbf{x})}\right|\leq c \right\}
&=\frac{1}{\sqrt{2\pi}}\int_{-c}^{c} e^{-\frac{\tau^2}{2}} d\tau \\
&=2\Phi(c)-1
\end{split}
\end{equation*}
where $\Phi$ denotes the cumulative distribution function (CDF) of the standard normal distribution. Since $\sigma^N_h(\mathbf{x})>0$, this indicates that
\begin{equation*}
\texttt{Prob}\left\{ \left|V(\mathbf{x})-\mu^N_h(\mathbf{x})\right|\leq c\cdot\sigma^N_h(\mathbf{x})\right\} =2\Phi(c)-1,
\end{equation*}
which is equivalent to
$$
\texttt{Prob}\left\{ \left|V(\mathbf{x})-\mu^N_h(\mathbf{x})\right|\leq \beta_{\delta}\cdot\sigma^N_h(\mathbf{x})\right\} =\delta
$$
with $\delta=2\Phi(c)-1$ and $\beta_{\delta}=\Phi^{-1}(\frac{1+\delta}{2})$. Therefore, $\mathbf{x}$ is in the level set $\{\mathbf{x}\in R^n|\mathbf{V}(\mathbf{x})\leq\hat{V}^N_{\max}\}$ with the probability of $\delta$ when it satisfies the inequality
$$
\mu^N_h(\mathbf{x})+\beta_{\delta}\cdot\sigma^N_h(\mathbf{x})\leq\hat{V}^N_{\max}
$$
with $\hat{V}^N_{\max}=\max_{\mathbf{x}^{(i)}\in W_h^N}\hat{V}(\mathbf{x}^{(i)})$. Considering that the level set is a compact subset of region of attraction $S$, the estimated region of attraction with the probability of $\delta$ is given by
$$
\Omega^N_{\delta}=\left\{\mathbf{x}\in R^n|\mu^N_h(\mathbf{x})+\beta_{\delta}\cdot\sigma^N_h(\mathbf{x})\leq\hat{V}^N_{\max}\right\}.
$$
This completes the proof.

\end{document}